\documentclass[pra,amsmath,amssymb,byrevtex]{revtex4}
\usepackage{amsmath,amssymb,amsthm,graphicx,color}

\newtheorem{Thm}{Theorem}

\newtheorem{Prop}{Proposition}
\newtheorem{Lem}{Lemma}

\newtheorem{Rem}{Remark}

\theoremstyle{definition}

\newcommand{\R}{\mathop{\mathbb{R}}\nolimits}

\newcommand{\N}{\mathop{\mathbb{N}}\nolimits}
\newcommand{\tr}{\mathop{\mathrm{Tr}}\nolimits}
\newcommand{\I}{\mathop{\mathbb{I}}\nolimits}

\DeclareMathOperator{\Tr}{Tr}

\newcommand{\ketbra}[2]{| #1 \rangle \langle #2 |}
\newcommand{\norm}[1]{|| #1 ||}  \newcommand{\cnorm}[1]{\left\|#1 \right\|}
\newcommand{\bignorm}[1]{\big\|#1\big\|}

\def\hil{\mathcal{H}}
\def\tc{\mathcal{T}}

\def\ki{\textit}
\def\bz{\left(}
\def\jz{\right)}
\def\d{d}
\def\dom{\mathcal{D}}
\def\ran{\mathcal{R}}
\def\U{\mathcal{U}}
\def\imp{\Longrightarrow}
\newcommand{\s}{\mbox{ }}
\newcommand{\ds}{\mbox{ }\mbox{ }}
\newcommand{\diad}[2]{|#1\rangle\langle #2|}
\newcommand{\proj}[1]{\diad{#1}{#1}}
\newcommand{\comm}[2]{[#1,#2]}
\newcommand{\inner}[2]{\langle #1 , #2\rangle}
\newcommand{\abs}[1]{\left|#1\right|}

\begin{document}

\title{Relation between the Dynamics of the Reduced Purity and Correlations}
\author{Gen Kimura $^a$}
\email{gen-kimura[at mark]aist.go.jp}
\author{Hiromichi Ohno $^b$}
\email{h_ohno[at mark]shinshu-u.ac.jp}
\author{Mil\'an  Mosonyi $^c{}^{,d}$}
\email{milan.mosonyi[at mark]gmail.com}
\affiliation{$^a$ Research Center for Information Security (RCIS), National Institute of Advanced Industrial Science and Technology (AIST). Daibiru building 1102, 1-18-13 Sotokanda, Chiyoda-ku, Tokyo, 101-0021, Japan.}
\affiliation{$^b$ Department of Mathematics, Faculty of Engineering, Shinshu University,
4-17-1 Wakasato, Nagano, 380-8553, Japan}
\affiliation{$^c$Mathematical Institute, Budapest University of Technology and Economics, Egry J$\acute{o}$zsef u 1., Budapest, 1111 Hungary}
\affiliation{$^d$ Centre for Quantum Technologies, National University of Singapore, Block S15, 3 Science Drive 2, Singapore 117543.}

\begin{abstract}
A general property of the relation between the dynamics of the reduced purity and correlations is investigated in quantum mechanical systems. 
We show that a non-zero time-derivative of the reduced purity of a system implies the existence of non-zero correlations with its environment under any unbounded Hamiltonians with finite variance.  
This shows the role of local dynamical information on the correlations, as well as the role of correlations in the mechanism of purity change. 
\end{abstract}

\maketitle

\section{Introduction}

In the theory of open quantum systems \cite{ref:OQS,ref:BP,ref:Davies}, an interaction between a system $S$ and its environment is essential to have a non-unitary time-evolution of the system. 
Indeed, without interaction, the system $S$ evolves unitarily even under the existence of its environment. 
In general, an interaction generates correlations between the system and its environment, from which the mechanism of the purity change of the system is explained. 
This follows from the general property of correlations: {\it the state of a system is mixed if there are correlations between the system and its environment}. 
Notice that the contraposition of this statement tells that (A) {\it if the system is in a pure state then it has no correlations with any other system} \footnote{See, for instance, \cite{ref:D} in which the physical importance of Statement (A) is clarified. Indeed, this is quite generic property of correlations, and one can show this in any general probabilistic theories: See, for instance [M. Takesaki, \textit{Theory of Operator Algebra I} (Springer, 1979)] and [H. Barnum, J. Barrett, M. Leifer, A. Wilce, Phys. Rev. Lett. 99, 240501 (2007); arXiv:0805.3553.] }. 
This simple but quite general property of correlations is as important as anything in the context of safe key distribution, such as quantum cryptography. 
In order to achieve a safe key distribution, the legitimate users, Alice and Bob, should prepare a pure (entangled) state of their system $S = A+B$. 
If an eavesdropper, say Eve, wants to get information on the system $S$, she can make her system $E$ interact with $S$. 
However, in order to get information on $S$, she has to create correlation between her system $E$ and $S$, and this would inevitably change the purity of the system $S$.  
Thus, what Alice and Bob should do is just to confirm that their system $S$ is indeed in a pure state, and this can be done in principle, provided that they have enough copies of their state (i.e., i.i.d. states). 

In essence, what we learn from Statement (A) is that a local information (the purity of $S$) can yield global information (correlations between $S$ and $E$). 
It is interesting to consider the following general problem: How much local information of a subsystem can have the global information of correlations? 
A short consideration, however, reveals that Statement (A) is all what one can learn about correlations from the local information. 
Indeed, if the (reduced) state is mixed $\rho_S$, then there could be both cases of zero correlations and non-zero correlations. 

However, it should be noticed that we used only a static property of the reduced state. 
Since we have i.i.d. states, we can measure local observables as much as possible, especially at each time where the reduced state evolves as time passes. 
One can expect to get some information of correlations using the local information of {\it the dynamics} of the subsystem.  
In \cite{ref:KOH}, we have shown the following: 
(B) {\it If the time evolution of the purity of a system has a non-zero time-derivative, then there are correlations between the system and its environment}. 
This seems natural but not so trivial: Notice that we know, from the role of an interaction, that a non-zero time-derivative of the reduced purity surely implies the existence of an interaction, but not directly the existence of correlations. 
To see the non-triviality of this, consider the contraposition of (B), which implies that, even if there is a strong interaction between $S$ and $E$, the purity of $S$ does not change at the moment of the time where there are no correlations. 
In this sense, one can say that correlations play an essential role in the change of purity.  

In order to see the Statement (B), we need to put some reasonable assumptions on the dynamics. 
In \cite{ref:KOH}, we (naturally) assumed that the total quantum system $S+E$ is isolated, and thus evolves under the Schr\"odinger equation with a total Hamiltonian. 
For technical simplicity, however, we have also assumed that the total Hamiltonian is bounded. 
In order to prove Statement (B) in full generality, we need to study the case where the total Hamiltonian is unbounded. 
In this note, we prove Statement (B) for unbounded Hamiltonians, provided that the total state has finite variance of the energy (total Hamiltonian). 
The finiteness of the variance is essential, since otherwise Statement (B) can fail in general (see \cite{ref:KOH} for a counterexample). 

As the von Neumann-Schr\"odinger equation plays an important role in our analysis, we give a brief review on its validity in section \ref{sec:DE}, where we show that finite variance of the energy in some state is a sufficient condition for the von Neumann-Schr\"odinger equation to hold.
In section \ref{sec:DofRP}, we prove our main results, including the above given statement and an upper bound on the time-derivative of the reduced purity in terms of the second moment of the total energy. 
After some discussions on the physical meaning of these results, we conclude the paper in section \ref{sec:C}. 

\section{Differentiable evolution}\label{sec:DE}

Consider a quantum mechanical system with separable Hilbert space $\hil$ and a Hamilton operator $H$. 
In the following, we denote the inner product between $\psi,\phi \in \hil$ by $\inner{\psi}{\phi}$, and adopt the convention that the inner product is linear in its second argument and anti-linear
in the first. 
The norm of a vector $\psi \in \hil$ is given by $||\psi||:= \inner{\psi}{\psi}^{1/2}$, and the norm of a bounded linear operator $A$ is defined as $||A|| := \sup\{||A\psi|| : \psi \in \hil, ||\psi || \le 1\}$.

In general, $H$ is an unbounded self-adjoint operator on $\hil$. 
States of the system are described by density operators, i.e., positive semidefinite trace-class operators on $\hil$ with unit trace.
If the system is isolated then its time evolution is determined by the unitary group $U(t):=e^{-itH},\s t\in\R$, generated by $H$.
More precisely, if the initial state of the system at time $t=0$ is $\rho$ then the state $\rho(t)$ of the system at time $t$ is given by 
\begin{equation}\label{time evolution}
\rho(t)=U(t)\rho U(t)^*=:\U_t\rho\,.
\end{equation}
One can easily verify that if $H$ is bounded then the evolution $t\mapsto\rho(t)$ is differentiable at any time $t$, and $\frac{d}{dt}\rho(t)=-i\comm{H}{\rho(t)}:=-i\bz H\rho(t)-\rho(t)H\jz$. 
However, when $H$ is unbounded, there exist initial states for which the evolution is not differentiable. 
Note that the space $\tc(\hil)$ of trace-class operators is a Banach space with the trace norm
$\norm{a}_1:=\Tr |a|$, and if $f:\,\R\to\tc(\hil)$ is a function then $\frac{d}{dt}\big|_{t=t_0}f(t)=a$ means $\lim_{t\to t_0}\bignorm{\frac{f(t)-f(t_0)}{t-t_0}-a}_1=0$. It is easy to see that $\lim_{t\to t_0}\norm{\U_t\rho-\U_{t_0}\rho}_1=0$ for any state $\rho$ and $t_0\in\R$, and $\U_{t+s}=\U_t\U_s,\,t,s\in\R$, i.e., $\U$ defines a strongly continuous group on $\tc(\hil)$. As a consequence, if the evolution is differentiable at some $t=t_0$ then it is differentiable at any $t\in\R$.
The following was given in \cite[Lemma 5.1]{ref:Davies}:
\begin{Thm}\label{Davies}
The map $t\mapsto\rho(t)$ is differentiable if and only if (i) $\rho$ leaves the domain of $H$ invariant, and (ii) $\comm{H}{\rho}$ is a densely defined closable operator such that its closure $\overline{\comm{H}{\rho}}$ is in $\tc(\hil)$.
Moreover, in this case
\begin{equation}\label{vNS}
\frac{d}{dt}\rho(t)=-i\overline{\comm{H}{\rho(t)}}
= -i U(t)\overline{\comm{H}{\rho}}U(t)^* \,,\ds\ds\ds t\in\R.
\end{equation}
\end{Thm}
\noindent We refer to \eqref{vNS} as the \ki{von Neumann-Schr\"odinger equation}. 

Though the above theorem gives a complete mathematical characterization of differentiability, its condition doesn't seem to have a direct physical interpretation.
Below we show a more physical sufficient condition, namely that the evolution is differentiable whenever the variance of the energy is finite in the initial state of the system.

Let $E^H(B)$ denote the spectral projection of $H$, corresponding to some Borel set $B\subset\R$. For any state $\rho$, the map $B\mapsto\Tr \bz E^H(B)\rho\jz$ defines a probability measure on $\R$, and 
the
$k$th \ki{moment} of the Hamiltonian in the state $\rho$ is defined as
\begin{equation*}
m_{k,\rho}(H):=\int_{\R}\lambda^k\,\d\Tr \bz E^H(\lambda)\rho \jz,
\end{equation*}
whenever the integral exists. Finiteness of some moment implies the finiteness of all lower moments (due to H\"older's inequality). In particular, if the second moment is finite then the \ki{expectation} and the \ki{variance} of $H$ with respect to $\rho$,
\begin{equation*}
E[H]_{\rho}:=\int_{\R}\lambda\,\d\Tr \bz E^H(\lambda)\rho \jz\ds\text{and}\ds
V[H]_{\rho}:=\int_{\R}\bz\lambda-E[H]_{\rho}\jz^2\,\d\Tr \bz E^H(\lambda)\rho \jz
\end{equation*}
are also finite. Note that the domain $\dom(H)$ of $H$ consists of those vectors $\psi$ for which the second moment is finite with respect to $\proj{\psi}$, and $\norm{H\psi}^2=m_{2,\proj{\psi}}(H), \ \psi \in \dom(H)$.
We will use the short-hand notation $V[H]_{\rho}<+\infty$ to indicate that the variance exists and is finite (which is easily seen to be equivalent to the finiteness of the second moment). 
Note that since $U(t)$ commutes with the spectral projections $E^H(B)$ for any Borel set $B \subset \R$ and any $t \in \R$, one easily obtains that  $V[H]_{\rho}=V[H]_{\rho(t)}$ and $m_{k,\rho}(H) = m_{k,\rho(t)}(H)$ for any $k \in \N, t \in \R$.

\begin{Prop}\label{prop:variance}
If the second moment of $H$ in the state $\rho$ is finite then the evolution is differentiable and the von Neumann-Schr\"odinger equation \eqref{vNS} holds.
\end{Prop}

\begin{Rem}
In this proposition, we can not replace the finiteness of the second moment with that of the first moment.
For example, let 
\[
\rho = \bigoplus_{n=1}^\infty
{1\over 2^n}\left(
\begin{array}{cc}
1 & 0 \\
0 & 0
\end{array} \right), \qquad
H = \bigoplus_{n=1}^\infty
 2^n \left(
\begin{array}{cc}
{1\over n^{2}} & {1\over n} \sqrt{1-{1\over n^2}} \\
{1\over n} \sqrt{1-{1\over n^2}} & 1-{1\over n^2}
\end{array} \right).
\]
Then the first moment of $H$ in the state $\rho$ is 
finite but $\overline{[H,\rho]}$ is not a trace-class operator.
\end{Rem}

The proof of Proposition \ref{prop:variance} is based on the following:
\begin{Lem}\label{lemma:vNS}
Assume that the range $\ran(\rho)$ of $\rho$ is contained in the domain of $H$ and 
\begin{equation*}
\sum_k p_k\norm{He_k}<+\infty\,,
\end{equation*}
where $\rho=\sum_k p_k\proj{e_k}$ is an eigen-decomposition of $\rho$, with $p_k>0,\,\sum_k p_k=1$. 
Then,
\begin{equation*}
H\rho=\sum_k p_k\diad{He_k}{e_k}\,,\ds\ds\ds 
\overline{\rho H}=\sum_k p_k\diad{e_k}{He_k}=(H\rho)^*\,,
\end{equation*}
where the sums converge in trace-norm, so that both $H \rho$ and $\overline{\rho H}$ are trace-class, and $[H,\rho]$ is closable with 
$\overline{[H,\rho]} = \sum_k p_k \ketbra{H e_k}{e_k} - \ketbra{e_k}{H e_k}$. 
Moreover, the state evolution satisfies the von Neumann-Schr\"odinger equation as
\begin{equation*}
\frac{d}{dt}\rho(t)=-i\sum_k p_kU(t)\bz \diad{He_k}{e_k}-\diad{e_k}{He_k}\jz U(t)^*\,.
\end{equation*}
\end{Lem}
\begin{proof}
Let $\rho_n:=\sum_{k=1}^n p_k\proj{e_k}$, and define 
$a_n:=\sum_{k=1}^n p_k\diad{He_k}{e_k}$,  $b_n:=a_n^\ast = \sum_{k=1}^n p_k\diad{e_k}{He_k}$ for all $n\in\N$. Note that by assumption, $e_k\in\dom(H)$ for all $k$, hence $a_n$ and $b_n$ are well-defined, and $a_n=H\rho_n$ and 
$b_n \psi=\rho_n H\psi,\,\psi\in\dom(H)$. Since $\bignorm{\diad{\psi}{\phi}}_1=\norm{\psi}\,\norm{\phi},\,\psi,\phi\in\hil$, we have 
\begin{equation*}
\sum_k \bignorm{p_k \diad{He_k}{e_k}}_1=\sum_k \bignorm{p_k \diad{e_k}{He_k}}_1=
\sum_k p_k \norm{He_k}<+\infty\,,
\end{equation*}
and thus the operators $a:=\sum_k p_k\diad{He_k}{e_k}$ and $b:=\sum_k p_k\diad{e_k}{He_k}$ are well-defined in $\tc(\hil)$, and, moreover, $\lim_n\norm{a-a_n}_1=\lim_n\norm{b-b_n}_1=0$. 
Taking the adjoint is a continuous operation with respect to the trace-norm, and hence, $b = \lim_n b_n = \lim_n a^\ast_n = a^\ast$. 

Let $\psi \in \hil$. 
By assumption, $\rho \psi$ is in $\dom(H)$ and, since $\rho_n$ converges to $\rho$ in trace-norm, $\rho_n \psi$ converges to $\rho\psi$. 
Moreover $H\rho_n \psi = a_n \psi$ converges to $a\psi$ by the above argument. 
Closedness of $H$ then yields that $H\rho \psi  = a\psi$.
Since this holds for all $\psi \in \hil$, we finally conclude that $
H\rho = a$. 
Similarly, for every $\psi\in\dom(H)$ we have $\norm{(\rho H-b)\psi}=\lim_n\norm{(\rho H-b_n)\psi}=\lim_n\norm{(\rho-\rho_n)H\psi}=0$, and therefore 
the restriction of $b$ onto $\dom(H)$ coincides with $\rho H$, from which $\overline{\rho H}=b$. 
The last assertion follows from Theorem \ref{Davies}.
\end{proof}
 
\textit{Proof of Proposition \ref{prop:variance}:}\s 
Since $\Tr\left(P\proj{\rho^{1/2}\psi}\right)=\inner{\psi}{\rho^{1/2}P\rho^{1/2}\psi}\le \Tr\left(\rho^{1/2}P\rho^{1/2}\right)=
\Tr \left(P\rho\right)$ for any unit vector $\psi\in\hil$ and projection $P$, we get
\begin{equation*}
\int_{\R}\lambda^2\,\d\Tr \left( E^H(\lambda)\proj{\rho^{1/2}\psi}\right)\le
\int_{\R}\lambda^2\,\d\Tr \left( E^H(\lambda)\rho\right)<+\infty\,,
\end{equation*}
and hence $\rho^{1/2}\psi\in\dom(H)$ for any unit vector $\psi\in\hil$, which in turn yields $\ran(\rho)\subset\dom(H)$. Similarly,
\begin{equation*}
\sum_k p_k\norm{He_k}^2= \sum_k p_k\int_{\R}\lambda^2\,\d\Tr \bz E^H(\lambda)\proj{e_k}\jz=\int_{\R}\lambda^2\,\d\Tr \bz E^H(\lambda)\rho \jz<+\infty\,,
\end{equation*}
and, by the Schwarz inequality,
\begin{equation*}
\sum_k p_k\norm{He_k}\le 
\bz\sum_k p_k \jz^{1/2}\bz\sum_k p_k \norm{He_k}^2\jz^{1/2} =
\bz\sum_k p_k\norm{He_k}^2\jz^{1/2}<+\infty.
\end{equation*}
By Lemma \ref{lemma:vNS}, the statement follows.\qed
\medskip

As is well-known (cf.~\cite{ref:Arai}), the Schr\"odinger equation holds for a Hamiltonian $H$ and an initial vector state $\proj{\psi}$ if and only if $\psi$ is in the domain of $H$, that is, if the second moment of $H$ is finite in the state $\proj{\psi}$. Proposition \ref{prop:variance} is therefore a natural generalization of this fact for mixed initial states. Note, however, that Lemma \ref{lemma:vNS} and Proposition \ref{prop:variance} only 
provide sufficient conditions for the von Neumann-Schr\"odinger equation to hold.
Consider, for instance, a faithful state $\rho$ with eigen-decomposition 
$\rho={6\over \pi^2}\sum_{k=1}^{\infty} \frac{1}{k^2}\proj{e_k}$ and define the Hamiltonian to be $H:=\rho^{-1/2}$. 
One can easily see that $\ran(\rho)\subset\dom(H)$ and $\comm{H}{\rho}$ is closable with $\overline{\comm{H}{\rho}}=0$, 
hence, by Theorem \ref{Davies}, the von Neumann equation holds. 
On the other hand, ${6\over \pi^2}\sum_{k=1}^{\infty}\frac{1}{k^2}\norm{H e_k}={\sqrt{6}\over \pi}\sum_{k=1}^{\infty}\frac{1}{k}=+\infty$.

\section{The dynamics of the reduced purity}\label{sec:DofRP}

Consider now a system $S$ coupled to some environment $E$, with a joint Hilbert space $\hil=\hil_S\otimes\hil_E$. For a state $\rho$ of the total system $S+E$, we denote by $\rho_S$ and $\rho_E$ its reductions to $\hil_S$ and $\hil_E$, respectively. We say that $\rho$ has no correlations if it is a product state, i.e., $\rho = \rho_S \otimes \rho_E$. 
Otherwise, we that say $\rho$ has non-zero correlations. 
It is easy to see that $\rho$ has no correlations if and only if any observables $A$ on $S$ and $B$ on $E$ are statistically independent with respect to $\rho$.

As it was shown in \cite{ref:KOH}, if the total system evolves according to a bounded Hamiltonian $H$, with a decomposition $H=H_S\otimes I_E+I_S\otimes H_E+H_{int}$, then
\begin{equation}\label{eq:QE}
\abs{\frac{d}{dt}P_S(t)}\le 4\sqrt{2} \,\norm{H_{int}}\,\sqrt{I(\rho(t))}\,,
\end{equation}
where $I(\rho(t)):=S\bz\rho_{S}(t) \jz+S\bz \rho_{E}(t) \jz-S\bz\rho(t)\jz$ is the \ki{mutual information} between the systems $S$ and $E$ in the state $\rho(t)$, and $P_S(t) := \Tr(\rho_{S}(t)^2)$ is the purity of the state $\rho(t)$. While for unbounded Hamiltonians one cannot expect such a bound to hold, the time-derivative of the reduced purity still reveals some information on the correlations contained in $\rho(t)$. 
Namely, if the derivative of the reduced purity is non-zero at some time $t$ then there are necessarily some correlations between the system and its environment. 
To show this, we start with the following:

\begin{Lem}\label{lem:pc}
If the von Neumann-Schr\"odinger equation holds then the reduced purity $P_S(t)$ is time-differentiable at any time, and 
$$
\frac{d}{dt} P_S(t)\big|_{t=t_0} =-2i \tr_{SE}\left(\rho_{S}(t_0)\otimes \I_E \overline{[H,\rho(t_0)]}\right).  
$$
\end{Lem}
\begin{proof} Since $||\rho_{S}(t)\otimes \I_E|| \le 1$ and $\overline{[H,\rho(t_0)]}$ is a trace-class operator by Theorem \ref{Davies}, it is easy to check that
\begin{align}\label{al}
& \left|\frac{P_S(t_0 + t)-P_S(t_0)}{t} + 2i \tr_S\left( \rho_{S}(t_0) \tr_E \overline{[H,\rho(t_0)]}\right) \right| \nonumber \\
&\ds\ds\ds\ds\ds\le 
\left| \frac{\tr_S\left( \left(\rho_{S}(t_0 + t) - \rho_{S}(t_0)\right)\rho_{S}(t_0 + t) \right)}{t }+i \tr_S\left( \rho_{S}(t_0 + t) \tr_E \overline{[H,\rho(t_0)]}\right) \right| \nonumber  \\
&\ds\ds\ds\ds\ds\ds\ds+
\left| {\tr_S\left(\rho_{S}(t_0) (\rho_{S}(t_0 + t) - \rho_{S}(t_0))\right)\over t} +i \tr_S\left( \rho_{S}(t_0) \tr_E \overline{[H,\rho(t_0)]}\right) \right| \nonumber  \\
&\ds\ds\ds\ds\ds\ds\ds+
\left| {\tr_S\left( (\rho_{S}(t_0 + t) - \rho_{S}(t_0)) \tr_E \overline{[H,\rho(t_0)]}\right)}\right| \nonumber  \\
&\ds\ds\ds\ds\ds\le
 2 \cnorm{ \frac{ \rho (t_0 + t)-\rho(t_0)}{t} + i\overline{[H,\rho(t_0)]}}_1
+
\cnorm{ {\rho_{S}( t_0 + t) - \rho_{S}( t_0) }} \cdot
\cnorm{  \tr_E  \overline{[H,\rho(t_0)]}}_1.
\end{align}
In the second step we used that $|\tr A B| \le ||A|| ||B||_1$ holds for any bounded operator $A$ and trace class operator $B$, and that the partial trace operation is continuous with respect to trace norm.

Thanks to Theorem \ref{Davies} and the strong continuity of $\U_t$, both terms in \eqref{al} go to $0$ as $t$ goes to $0$. 
\end{proof}
Note that under the conditions of Lemma \ref{lemma:vNS}, one has $\overline{H\rho-\rho H}=H\rho-\bz H\rho\jz^*$, and hence, by Lemma \ref{lem:pc},
\begin{equation*} 
\abs{\frac{d}{dt} P_S(t)\big|_{t=t_0}}\le
4 \norm{\rho_S(t_0)}\,\norm{H\rho(t_0)}_1
\le 4 \sum_k p_k \norm{H e_k}\,,
\end{equation*}
where $\rho = \sum_{k=1}^\infty p_k \ketbra{e_k}{e_k}$ is an eigen-decomposition. 
If, moreover, $V[H]_{\rho}<+\infty$ then the conditions of Lemma \ref{lemma:vNS} hold, 
and the above can further be upper bounded by $4 m_{2,\rho}(H)^{1/2}$ as it was shown in the proof of Proposition \ref{prop:variance}. 
Hence, we obtain an upper bound on the change of the reduced purity in terms of the expectation value of the square of the Hamiltonian: 
\begin{Prop}\label{prop:UB} If $V[H]_{\rho} < \infty \ (\Leftrightarrow m_{2,\rho}(H) < \infty)$, we have 
\begin{equation}\label{eq:QEmoment} 
\abs{\frac{d}{dt} P_S(t)\big|_{t=t_0}}\le 
4 m_{2,\rho}(H)^{1/2}. 
\end{equation}
\end{Prop}
\bigskip 

Now we are in a position to give our main result: 
\begin{Thm}\label{Thm:main} 
Assume that the total system $S+E$ evolves according to the Hamiltonian $H$, with an initial state $\rho_0$, such that 
$V[H]_{\rho_0} < \infty $. If there is no correlation between the system and its environment at some time $t_0$ then 
the reduced purity has a flat derivative at that moment, i.e.,
\begin{equation*}
\rho(t_0)=\rho_{S}( t_0) \otimes \rho_{E} (t_0) \ds\imp\ds \frac{d}{dt} P_S(t)\big|_{t=t_0} = 0. 
\end{equation*}
\end{Thm}
\begin{proof}
From Proposition \ref{prop:variance} and its proof,
the assumption of finite variance implies the assumptions of 
Lemmas \ref{lemma:vNS} and \ref{lem:pc}.
Assume that $\rho(t_0)=\rho_{S}( t_0) \otimes \rho_{E}( t_0)$ at some $t_0$.  
To simplify notation, let $\rho:=\rho(t_0)$, $\rho_S := \rho_{S}( t_0)$ and $\rho_E := \rho_{E}( t_0)$. 
Let $\rho_S=\sum_k p_k\proj{\psi_k}$ and $\rho_E=\sum_l q_l\proj{\phi_l}$ be eigen-decompositions of $\rho_S$ and $\rho_E$, respectively, with all $p_k,q_l>0$. Then,
$\rho=\sum_{k,l}p_kq_l\proj{\psi_k\otimes\phi_l}$ is an eigen-decomposition of $\rho$,
and $\psi_k\otimes\phi_l\in\dom(H)$ for all $k,l$, by the assumption of finite variance.
Lemmas \ref{lemma:vNS} and \ref{lem:pc} give
\begin{eqnarray*}
\frac{d}{dt} P_S(t)\big|_{t=t_0} &=&-2i \tr_{SE}\left(\bz\rho_S\otimes \I_E\jz \overline{[H,\rho]}\right)\\
&=&-2i\Tr_{SE}\left(
\bz\rho_S\otimes \I_E\jz \sum_{k,l} \bz \diad{H\psi_k\otimes\phi_l}{\psi_k\otimes\phi_l}-\diad{\psi_k\otimes\phi_l}{H\psi_k\otimes\phi_l}\jz\right)\\
&=&-2i\sum_{k,l} \Tr_{SE}\bz\diad{H\psi_k\otimes\phi_l}{\bz\rho_S\psi_k\jz\otimes\phi_l} - \diad{\bz\rho_S\psi_k\jz\otimes\phi_l}{H\psi_k\otimes\phi_l}\jz\\\
&=&-2i\sum_{k,l} p_k\bz\inner{\psi_k\otimes\phi_l}{H\psi_k\otimes\phi_l}-\inner{H\psi_k\otimes\phi_l}{\psi_k\otimes\phi_l}\jz\\
&=&0\,.
\end{eqnarray*}
In the third step we used that for a fixed bounded operator $A$, the map $B\mapsto \Tr AB$ is continuous on $\tc(\hil)$ with respect to the trace-norm.
\end{proof}

\section{Concluding Remarks}\label{sec:C}

In this short note, we have investigated Statement (B) in quantum mechanical systems with arbitrary Hamiltonians described by possibly unbounded self-adjoint operators.  
Theorem \ref{Thm:main} implies that the property of correlation in Statement (B) is universal if one accepts the finiteness of the variance of the total energy. 
Statement (B) has several suggestions on the general property of correlations: 

First, the contraposition of this reveals the role of correlation in the mechanism of purity changes. Namely, we {\it need not only an interaction but also correlations} to change the reduced purity. 

Second, Statement (B) can be used as a method to detect correlations from local information: 
Namely, if one finds that a system's purity has a non-zero time-derivative at a certain time instance then one can conclude that the system has non-zero correlations with its environment at that time instance. Unlike in the case of bounded Hamiltonians,it seems difficult to obtain a quantitative estimation such as \eqref{eq:QE}. 
Instead, we have shown an upper bound in Proposition \ref{prop:UB} on the time-derivative of the reduced purity in terms of the second moment of the total Hamiltonian. 
It would be interesting to find a quantitative version of Statement (B), such that the time-derivative of the purity is bounded by a certain quantity including both the magnitudes of correlations and the second moment (or the variance) of the total Hamiltonian. 

Third, Statement (B) reveals a physical difference between the {\it proper mixture} and the {\it improper mixture}. 
In \cite{ref:D}, two origins of mixture of quantum states are conceptually distinguished.   
A mixture of a state is said to be {\it proper} if the origin of the mixture is due to the absence of knowledge; 
for instance, if a state of a system $S$ is prepared in one of the pure states $\{\psi_i\}_i$ with a prior probability $p_i$, then we describe the state by the mixed density operator $\rho_S = \sum_i p_i \ketbra{\psi_i}{\psi_i}$. (Note that, in ``reality", the state is one of the pure state $\psi_i$, but it is reasonable in a statistical sense to represent the state by the mixed density operator $\rho_S$ if one does not have knowledge about which state was prepared.) 
On the other hand, based on Statement (A), there exists another origin of mixture due to correlations with environment $E$:  
Even if a state of the total system $S+E$ is in a pure state $\psi$, the reduced state is a mixed state $\tilde{\rho}_S = \tr_E \ketbra{\psi}{\psi}$ if $\psi$ has correlations between $S$ and $E$. 
A mixture of this kind, i.e., with the origin due to correlations, is said to be {\it improper}. 
As a general agreement, however, both states $\rho_S$ and $\tilde{\rho}_S$ are physically identified if they are described by the same density operator, and thus the difference between properness and improperness of mixture is just of conceptual one, and there does not apper any physical difference between them.  
On the other hand, from Statement (B), if one finds a non-zero time-derivative of the purity of the system $S$, then there must exist correlations with its environment. Therefore, in such case, one can confirm that the mixture of the state is due to correlations, and thus is improper. 
Indeed, since a proper mixed state is not originate from correlations, and thus we can assume that the total state is of the form $\rho_{tot} = \rho_S\otimes \rho_E$ for a proper mixed state $\rho_S$. 
Therefore, a time-derivative of the purity of the system $S$ cannot be non-zero in a proper case. 
In this sense, there could be a physical diffrence between proper and improper mixed states. 
Note that, however, a mixed state would be in general a hybrid of proper and improper mixtures.

Finally, let us conclude our paper with a perspective of a measure for (quantum) correlations using local information. 
When the state of the total system is pure, the reduced entropy $H(\rho_S):=-\Tr_S (\rho_S\log\rho_S)$ is known to be a good measure of correlations between the system and its environment \cite{ref:EM}. This does not hold, however, if the total system is in a mixed state. 
For instance, if $\dim\hil_S=\dim\hil_E=d<+\infty$ then the maximally entangled pure state and the maximally mixed state both have reduced entropy equal to $\log d$. One can consider various other local information quantities to capture correlations, like the R\'enyi entropies $S_\alpha(\rho_S):=\frac{1}{1-\alpha}\log\Tr\rho_S^{\alpha},\,\alpha>0,\alpha\ne 1$, or the purity
$P_S(\rho_S)=\Tr\rho_S^2$,
which is the non-logarithmic version of the $2$-R\'enyi entropy. 
These quantities approximate the von Neumann entropy, and hence give useful measures of correlations when the total state is pure, but suffer from the same problem in the mixed state case. 
In spite of this, our result shows that local dynamical quantities can still give useful information on the amount of correlations, even if the total system may be in a mixed state. 
In this line, we will further investigate the role of the dynamical information of the reduced purity as a measure of correlations in the near future.      

\bigskip

{\bf Acknowledgement}
\bigskip

Most of this work was done when the authors belonged to Tohoku University in 2007, and we are grateful to Prof. M.~Ozawa, Prof. F.~Hiai, and Dr. T.~Miyamoto for their helpful comments and advices. 
G.K. would like to acknowledge Dr. K. Imafuku for his useful comments on the application of Statement (B) to a difference between proper and improper mixtures. 
Financial supports by the JSPS research grants (G.~K. and H.~O.), the Grant-in-Aid for JSPS fellows $18 \cdot 06916$ (M.M.), and the Hungarian Research Grant OTKA  T068258 (M.M.) are gratefully acknowledged. 
The Centre for Quantum Technologies is funded by the Singapore Ministry of Education and the National Research Foundation as part of the Research Centres of Excellence program.

\end{document}